\documentclass[pra,aps,nopacs,onecolumn,twoside,superscriptaddress]{revtex4}

\usepackage{amsmath,amsfonts,amssymb,caption,color,epsfig,graphics,graphicx,hyperref,latexsym,mathrsfs,revsymb,theorem,url,verbatim,epstopdf,mathtools,enumerate}

\hypersetup{colorlinks,linkcolor={blue},citecolor={blue},urlcolor={red}}

\newtheorem{definition}{Definition}
\newtheorem{proposition}[definition]{Proposition}
\newtheorem{lemma}[definition]{Lemma}

\newtheorem{theorem}[definition]{Theorem}
\newtheorem{corollary}[definition]{Corollary}
\newtheorem{conjecture}[definition]{Conjecture}

\newtheorem{remark}[definition]{Remark}
\newtheorem{example}[definition]{Example}
\newtheorem{question}[definition]{Question}
\newtheorem{memo}[definition]{Memo}

\def\squareforqed{\hbox{\rlap{$\sqcap$}$\sqcup$}}
\def\qed{\ifmmode\squareforqed\else{\unskip\nobreak\hfil
\penalty50\hskip1em\null\nobreak\hfil\squareforqed
\parfillskip=0pt\finalhyphendemerits=0\endgraf}\fi}
\def\endenv{\ifmmode\;\else{\unskip\nobreak\hfil
\penalty50\hskip1em\null\nobreak\hfil\;
\parfillskip=0pt\finalhyphendemerits=0\endgraf}\fi}
\newenvironment{proof}{\noindent \textbf{{Proof.~} }}{\qed}
\def\Dbar{\leavevmode\lower.6ex\hbox to 0pt
{\hskip-.23ex\accent"16\hss}D}
\makeatletter
\def\url@leostyle{%
  \@ifundefined{selectfont}{\def\UrlFont{\sf}}{\def\UrlFont{\small\ttfamily}}}
\makeatother
\def\bcj{\begin{conjecture}}
\def\ecj{\end{conjecture}}
\def\bcr{\begin{corollary}}
\def\ecr{\end{corollary}}
\def\bd{\begin{definition}}
\def\ed{\end{definition}}
\def\bea{\begin{eqnarray}}
\def\eea{\end{eqnarray}}
\def\beq{\begin{equation}}
\def\eeq{\end{equation}}
\def\bal{\begin{aligned}}
\def\eal{\end{aligned}}
\def\bem{\begin{enumerate}}
\def\eem{\end{enumerate}}
\def\bex{\begin{example}}
\def\eex{\end{example}}
\def\bim{\begin{itemize}}
\def\eim{\end{itemize}}
\def\bl{\begin{lemma}}
\def\el{\end{lemma}}
\def\bma{\begin{bmatrix}}
\def\ema{\end{bmatrix}}
\def\bpf{\begin{proof}}
\def\epf{\end{proof}}
\def\bpp{\begin{proposition}}
\def\epp{\end{proposition}}
\def\bqu{\begin{question}}
\def\equ{\end{question}}
\def\br{\begin{remark}}
\def\er{\end{remark}}
\def\bt{\begin{lemma}}
\def\et{\end{lemma}}
\def\bmm{\begin{memo}}
\def\emm{\end{memo}}

\def\btb{\begin{tabular}}
\def\etb{\end{tabular}}

\newcommand{\nc}{\newcommand}

\nc{\bbA}{\mathbb{A}} \nc{\bbB}{\mathbb{B}} \nc{\bbC}{\mathbb{C}}
 \nc{\bbD}{\mathbb{D}} \nc{\bbE}{\mathbb{E}} \nc{\bbF}{\mathbb{F}}
 \nc{\bbG}{\mathbb{G}} \nc{\bbH}{\mathbb{H}} \nc{\bbI}{\mathbb{I}}
 \nc{\bbJ}{\mathbb{J}} \nc{\bbK}{\mathbb{K}} \nc{\bbL}{\mathbb{L}}
 \nc{\bbM}{\mathbb{M}} \nc{\bbN}{\mathbb{N}} \nc{\bbO}{\mathbb{O}}
 \nc{\bbP}{\mathbb{P}} \nc{\bbQ}{\mathbb{Q}} \nc{\bbR}{\mathbb{R}}
 \nc{\bbS}{\mathbb{S}} \nc{\bbT}{\mathbb{T}} \nc{\bbU}{\mathbb{U}}
 \nc{\bbV}{\mathbb{V}} \nc{\bbW}{\mathbb{W}} \nc{\bbX}{\mathbb{X}}
 \nc{\bbZ}{\mathbb{Z}}

 \nc{\bA}{{\bf A}} \nc{\bB}{{\bf B}} \nc{\bC}{{\bf C}}
 \nc{\bD}{{\bf D}} \nc{\bE}{{\bf E}} \nc{\bF}{{\bf F}}
 \nc{\bG}{{\bf G}} \nc{\bH}{{\bf H}} \nc{\bI}{{\bf I}}
 \nc{\bJ}{{\bf J}} \nc{\bK}{{\bf K}} \nc{\bL}{{\bf L}}
 \nc{\bM}{{\bf M}} \nc{\bN}{{\bf N}} \nc{\bO}{{\bf O}}
 \nc{\bP}{{\bf P}} \nc{\bQ}{{\bf Q}} \nc{\bR}{{\bf R}}
 \nc{\bS}{{\bf S}} \nc{\bT}{{\bf T}} \nc{\bU}{{\bf U}}
 \nc{\bV}{{\bf V}} \nc{\bW}{{\bf W}} \nc{\bX}{{\bf X}}
 \nc{\bZ}{{\bf Z}}

\nc{\cA}{{\cal A}} \nc{\cB}{{\cal B}} \nc{\cC}{{\cal C}}
\nc{\cD}{{\cal D}} \nc{\cE}{{\cal E}} \nc{\cF}{{\cal F}}
\nc{\cG}{{\cal G}} \nc{\cH}{{\cal H}} \nc{\cI}{{\cal I}}
\nc{\cJ}{{\cal J}} \nc{\cK}{{\cal K}} \nc{\cL}{{\cal L}}
\nc{\cM}{{\cal M}} \nc{\cN}{{\cal N}} \nc{\cO}{{\cal O}}
\nc{\cP}{{\cal P}} \nc{\cQ}{{\cal Q}} \nc{\cR}{{\cal R}}
\nc{\cS}{{\cal S}} \nc{\cT}{{\cal T}} \nc{\cU}{{\cal U}}
\nc{\cV}{{\cal V}} \nc{\cW}{{\cal W}} \nc{\cX}{{\cal X}}
\nc{\cZ}{{\cal Z}}

\nc{\hA}{{\hat{A}}} \nc{\hB}{{\hat{B}}} \nc{\hC}{{\hat{C}}}
\nc{\hD}{{\hat{D}}} \nc{\hE}{{\hat{E}}} \nc{\hF}{{\hat{F}}}
\nc{\hG}{{\hat{G}}} \nc{\hH}{{\hat{H}}} \nc{\hI}{{\hat{I}}}
\nc{\hJ}{{\hat{J}}} \nc{\hK}{{\hat{K}}} \nc{\hL}{{\hat{L}}}
\nc{\hM}{{\hat{M}}} \nc{\hN}{{\hat{N}}} \nc{\hO}{{\hat{O}}}
\nc{\hP}{{\hat{P}}} \nc{\hR}{{\hat{R}}} \nc{\hS}{{\hat{S}}}
\nc{\hT}{{\hat{T}}} \nc{\hU}{{\hat{U}}} \nc{\hV}{{\hat{V}}}
\nc{\hW}{{\hat{W}}} \nc{\hX}{{\hat{X}}} \nc{\hZ}{{\hat{Z}}}

\nc{\hn}{{\hat{n}}}

\newcommand{\bra}[1]{\langle#1|}
\newcommand{\ket}[1]{|#1\rangle}

\newcommand{\opp}{\red{OPEN PROBLEMS}.~}

\newcommand{\red}{\textcolor{red}}

\def\Dbar{\leavevmode\lower.6ex\hbox to 0pt
{\hskip-.23ex\accent"16\hss}D}

\begin{document}


\title{Markov Constraints Enhance Identifiability in Quantum Shadow Inversion}

\author{Zhixing Chen}\email[]{2409108@buaa.edu.cn}
\affiliation{LMIB(Beihang University), Ministry of education, and School of Mathematical Sciences, Beihang University, Beijing 100191, China} 

\author{Lin Chen}\email[]{linchen@buaa.edu.cn (corresponding author)}
\affiliation{LMIB(Beihang University), Ministry of education, and School of Mathematical Sciences, Beihang University, Beijing 100191, China}

\begin{abstract}
We study quantum shadow inversion under Markovian locality constraints
for four-partite systems arranged along the chain $A$--$B$--$C$--$D$.
The goal is to reproduce the expectation value of a fixed endpoint
observable $O_{AD}$ after an unknown global unitary, without requiring
full unitary inversion. We formulate the task using Markov-admissible
supermaps and introduce the Markov-implementable centralizer to describe
the remaining endpoint gauge freedom. We show that unrestricted
endpoint post-processing is too broad, and impose an endpoint-local
refinement. Under this condition, every implementable endpoint unitary
must factorize across $A|D$, so the Markov constraint strictly reduces
the centralizer-induced shadow ambiguity whenever the full centralizer
contains non-product unitaries. This provides a structural mechanism by
which Markov locality enhances identifiability in quantum shadow
inversion.
\end{abstract}

\maketitle



\section{Introduction}
\label{sec:intro}

Quantum unitary inversion is a fundamental primitive in quantum information processing, with applications in quantum metrology, quantum circuit verification, and quantum control \cite{1,2}. In its full form, one seeks a protocol that implements the inverse channel \(\mathcal{U}^\dagger\) for any unknown unitary \(U\). However, in many practical scenarios, it suffices to reproduce only the expectation value of a fixed observable, rather than the entire unitary action. This relaxed goal is known as \emph{shadow unitary inversion} \cite{Aaronson2018,Huang2020}, and it offers a resource-efficient alternative when global reconstruction is infeasible due to limited coherence or measurement capabilities \cite{Kliesch2021,Ducuara2020}.

The theory of quantum supermaps, introduced by Chiribella et al. \cite{1,2}, provides the natural framework for describing protocols that transform quantum channels. In particular, a one-slot quantum supermap (also called a 1-comb) is represented by a positive operator that maps the Choi operator of an input channel to that of an output channel, as will be presented in Example~\ref{ex:one_slot_supermap}. This framework has been extended to quantum networks with memory, enabling the study of adaptive protocols and causal structures \cite{Gutoski2007}.

In multipartite settings, the observable of interest may act on a subset of the total system. For a four-partite system with subsystems \(A,B,C,D\) arranged along a chain, one may only care about an endpoint observable \(O_{AD}\) acting on \(\mathcal{H}_A \otimes \mathcal{H}_D\). The goal is then to design a protocol that reproduces the expectation value of \(O_{AD}\) after an unknown global unitary \(U_{ABCD}\), without requiring full unitary inversion. This leads naturally to the notion of \emph{four-body shadow Markov inversion}, where the protocol is constrained to respect a Markovian locality structure along the chain, so that encoding and decoding operations act only locally with possible memory, but do not allow arbitrary global processing \cite{Brandao2015,Buscemi2017}. This architecture is inspired by recent work on quantum Markov chains and local operations with bounded communication \cite{Kretschmann2005,Perez-Delgado2007}.

A central difficulty in shadow inversion is the intrinsic ambiguity caused by the observable's symmetry. If \(O_{AD}\) has degenerate eigenvalues, its centralizer
\[
C_{O_{AD}} = \{ V_{AD} : [V_{AD}, O_{AD}] = 0 \}
\]
is nontrivial, and any unitary \(V_{AD}\) in the centralizer leaves the expectation value of \(O_{AD}\) unchanged. Thus shadow inversion can only recover the global unitary up to equivalence classes modulo \(C_{O_{AD}}\) \cite{Navascues2015,Skrzypczyk2010}. As shown in Lemma~\ref{lem:non_product_centralizer}, such centralizers can contain genuinely non-product unitaries for observables like \(Z_A\otimes Z_D\) on GHZ or W states, and Lemma~\ref{lem:centralizer_shadow_equivalence} further establishes that any centralizer element generates a shadow-equivalence class. This ambiguity is well understood in unconstrained settings and tripartite cases \cite{Chen2024,Liu2024}, but the four-partite chain case has not been systematically analyzed, particularly regarding how the causal structure of the protocol interacts with the centralizer-induced equivalence.

In this paper, we address this gap by formulating the task using Markov-admissible supermaps in Definition~\ref{def:markov_admissible_supermap}. The well-definedness of this realization is established in Proposition~\ref{prop:markov_realization_well_defined}, and the shadow correctness condition is characterized in Eq.~\eqref{eq:shadow_condition_recalled}. We introduce the Markov-implementable centralizer \(C_{O_{AD}}^{\mathrm{Markov}}\) in Eq.~\eqref{eq:markov_centralizer} to capture the endpoint gauge transformations realizable under the Markov constraint, and define the four-body Markov shadow inversion protocol in Definition~\ref{def:markov_shadow_inversion}. We show that unrestricted endpoint post-processing is too broad in Proposition~\ref{prop:endpoint_postprocessing_closure} and therefore impose an endpoint-local refinement in Definition~\ref{def:endpoint_local_markov}. Under this condition, we prove in Theorem~\ref{thm:main_endpoint_local_reduction} that every implementable endpoint unitary must factorize across \(A|D\). Consequently, whenever the full centralizer contains non-product unitaries, the Markov constraint strictly reduces the shadow-equivalence class in Corollary~\ref{cor:strict_shadow_reduction}, thereby enhancing identifiability. We illustrate our findings with four-qubit GHZ and W states, and formulate the feasibility of admissible protocols as a semidefinite program in Eq.~\eqref{eq:sdp_feasibility_markov_shadow}.

The rest of the paper is organized as follows. Section~\ref{sec:pro setup} introduces the problem setup and the shadow inversion condition. Section~\ref{subsec:markov_admissible_shadow} introduces the Markov-admissible architecture, the centralizer formalism, and the SDP formulation. Section~\ref{subsec:endpoint_local_reduction} proves the main theorem on endpoint-local reduction, and Section~\ref{sec:conclusion} concludes with discussions and open problems.

\section{Problem Setup}
\label{sec:pro setup}

We study shadow unitary inversion, a relaxed form of
unitary reversal in which the goal is not to fully invert an unknown
unitary $U$, but only to reproduce its effect on a fixed observable
$O$.

A quantum protocol $\mathcal N$ is called a $t$-query shadow inversion
protocol for $d$-dimensional unitaries with respect to an observable
$O$ if it is allowed to query an unknown unitary $U\in U(d)$ exactly
$t$ times. Here, ``querying'' $U$ means applying the unitary channel
$\mathcal{U}(\cdot)=U(\cdot)U^\dagger$ to a quantum state; in a circuit
picture, this corresponds to inserting the unknown unitary gate into
$t$ distinct slots. The number $t$ is fixed in advance, and the
protocol contains no measurement-based feedback that would alter the
number of applications of $U$. The protocol satisfies
\begin{equation}
\operatorname{Tr}\!\left[\mathcal N_U(\rho)O\right]
=
\operatorname{Tr}\!\left[U^\dagger \rho U O\right],
\qquad
\forall\,\rho,\quad \forall\,U\in U(d).
\label{eq:t_query_shadow_inversion}
\end{equation}
This condition will be referred to as the shadow inversion condition.
This is a natural relaxation in settings where only partial (shadow) information about the system is required.

In many multipartite quantum information tasks, one is not required to fully
recover an unknown global unitary evolution, but only to reproduce its action
on certain reduced observables or conditional subsystems.
This motivates a relaxed notion of unitary inversion that is compatible with
Markovian structures and shadow information.
In particular, for four-partite systems, such a relaxation should go beyond a
straightforward extension of the tripartite case, and explicitly capture the
independent roles of the four subsystems.

Let $A,B,C,D$ be four quantum systems with Hilbert spaces
$\mathcal{H}_A,\mathcal{H}_B,\mathcal{H}_C,\mathcal{H}_D$, and let
$O_{AD}$ be a fixed observable acting on $\mathcal{H}_A \otimes \mathcal{H}_D$. A quantum supermap \cite{1} $\mathcal{N}$ is called a \emph{four-body shadow Markov inversion}
with respect to $O_{AD}$ if, for any unitary $U_{ABCD}$ acting on
$\mathcal{H}_A \otimes \mathcal{H}_B \otimes \mathcal{H}_C \otimes \mathcal{H}_D$
and for any input state $\rho_{ABCD}$, the induced channel
$\mathcal{N}_U$ satisfies
\begin{equation}
\operatorname{Tr}\!\left[
\mathcal{N}_U(\rho_{ABCD})\, O_{AD}
\right]
=
\operatorname{Tr}\!\left[
U_{ABCD}^\dagger \rho_{ABCD} U_{ABCD}\, O_{AD}
\right].
\label{eq:four_body_shadow_inversion}
\end{equation}

\begin{example}
\label{ex:one_slot_supermap}
A one-slot quantum supermap \cite{2} (also called a 1-comb) is represented by
a positive operator $W$ acting on
$\mathcal{H}_{\text{out}} \otimes \mathcal{H}_{\text{in}}$,
such that for any input channel $\mathcal{E}$ with Choi operator $J_{\mathcal{E}}$,
the output channel $\mathcal{S}(\mathcal{E})$ has Choi operator
\begin{equation}
J_{\mathcal{S}(\mathcal{E})}
=
\operatorname{Tr}_{\text{in}}
\bigl[
W (J_{\mathcal{E}} \otimes I)
\bigr].
\end{equation}
\end{example}

Equation \eqref{eq:four_body_shadow_inversion} does not require
$\mathcal{N}_U(\rho_{ABCD}) = U_{ABCD}^\dagger \rho_{ABCD} U_{ABCD}$.
Instead, it enforces correctness only at the level of the observable $O_{AD}$,
thereby defining a shadow notion of unitary inversion.
All information orthogonal to $O_{AD}$ is allowed to be distorted. Here, orthogonal information refers to all state or operator components
that make no contribution to the expectation value of the observable $O_{AD}$,
i.e., all $X$ such that $\operatorname{Tr}(X O_{AD}) = 0$.

Multipartite entangled states with high symmetry provide a natural testing
ground for the concept of shadow Markov inversion.
In particular, four-qubit GHZ and W states represent two inequivalent classes
of genuine multipartite entanglement, characterized by fundamentally different
correlation and symmetry structures.
In this section, we analyze how these differences manifest themselves in the
existence and structure of four-body shadow Markov inversion protocols.

We begin with the four-qubit GHZ state
\begin{equation}
\ket{\mathrm{GHZ}} = \frac{1}{\sqrt{2}}\bigl(\ket{0000} + \ket{1111}\bigr).
\end{equation}
We consider an observable $O_{AD}$ acting on the subsystems $A$ and $D$ diagonal
in the computational basis, for example, $O_{AD} = Z_A \otimes Z_D$. Here $O_{AD}$ denotes a fixed Hermitian observable acting on
$\mathcal H_A \otimes \mathcal H_D$. The reduced state of $\ket{\mathrm{GHZ}}$ on subsystems $AD$ is supported on the
two-dimensional subspace spanned by $\{\ket{00}, \ket{11}\}$.
As a consequence, the observable $O_{AD}$ exhibits spectral degeneracy on the
relevant support, leading to a large centralizer $C_{O_{AD}} = \{V_{AD} : [V_{AD}, O_{AD}] = 0\}$. 
This enlarged symmetry implies that many distinct global unitary evolutions
$U_{ABCD}$ become indistinguishable at the level of the shadow observable
$O_{AD}$. Consequently, the condition \eqref{eq:four_body_shadow_inversion}
admits a broad family of solutions. In particular, a shadow Markov inversion can be realized without uniquely recovering the action of $U_{ABCD}$ on the full four-partite system.

More precisely, two global unitaries $U_{ABCD}$ and $U'_{ABCD}$ are
indistinguishable with respect to the shadow observable $O_{AD}$ if
\begin{equation}
\operatorname{Tr}\!\left[
U_{ABCD}^\dagger \rho_{ABCD} U_{ABCD}\, O_{AD}
\right]
=
\operatorname{Tr}\!\left[
U_{ABCD}'^\dagger \rho_{ABCD} U_{ABCD}'\, O_{AD}
\right]
\quad \forall \rho_{ABCD}.
\end{equation}
For observables with degenerate spectrum on the relevant support,
this defines a nontrivial equivalence class of unitaries modulo
the centralizer $C_{O_{AD}}$.

This highlights a fundamental distinction between full unitary inversion
and shadow Markov inversion.
By full unitary inversion we mean the reconstruction of the entire
adjoint action of $U_{ABCD}$, namely the requirement that the recovery
map $\mathcal N_U$ satisfies
\begin{equation}
\mathcal N_U(\rho_{ABCD}) = U_{ABCD}^\dagger \rho_{ABCD} U_{ABCD}
\quad \forall \rho_{ABCD}.
\end{equation}
While full unitary inversion requires the reconstruction of the complete action
of $U_{ABCD}$ on $\mathcal H_A \otimes \mathcal H_B \otimes
\mathcal H_C \otimes \mathcal H_D$, the shadow condition only constrains
expectation values of a fixed observable $O_{AD}$.
As a result, the recovery map need not reproduce the full unitary evolution,
but only its projection onto the algebra generated by $O_{AD}$.

In the four-qubit GHZ state, the reduced state on $AD$ is supported on the
two-dimensional subspace $\mathrm{span}\{\ket{00},\ket{11}\}$.
For diagonal observables $O_{AD}$, this induces spectral degeneracy on the
relevant support and hence a nontrivial centralizer
$C_{O_{AD}} = \{V_{AD} : [V_{AD}, O_{AD}] = 0\}$.

As a consequence, for any $V_{AD}\in C_{O_{AD}}$ and any global unitary
$U_{ABCD}$, the unitaries $U_{ABCD}$ and $(V_{AD}\otimes I_{BC})U_{ABCD}$
are shadow-equivalent.
Therefore, shadow Markov inversion recovers the global evolution only up to
an equivalence class modulo $C_{O_{AD}}$, rather than a unique unitary.

\begin{example}[A simple shadow-equivalent pair]
Consider the four-qubit GHZ state
\begin{equation}
\ket{\mathrm{GHZ}} = \frac{1}{\sqrt{2}}(\ket{0000}+\ket{1111}).
\end{equation}
We consider the shadow observable $O_{AD}=Z_A\otimes Z_D$.
Let
\begin{equation}
V_{AD} = e^{i\theta Z_A}\otimes I_D,
\end{equation}
which satisfies $[V_{AD},O_{AD}]=0$, and hence $V_{AD}\in C_{O_{AD}}$.

For an arbitrary global unitary $U_{ABCD}$, define
\begin{equation}
U'_{ABCD} := (V_{AD}\otimes I_{BC})\,U_{ABCD}.
\end{equation}
Then, for any input state $\rho_{ABCD}$, we have
\begin{equation}
\operatorname{Tr}\!\left[
U_{ABCD}\,\rho_{ABCD}\,U_{ABCD}^\dagger\, O_{AD}
\right]
=
\operatorname{Tr}\!\left[
U'_{ABCD}\,\rho_{ABCD}\,U_{ABCD}'^\dagger\, O_{AD}
\right].
\end{equation}
Therefore, $U_{ABCD}$ and $U'_{ABCD}$ are shadow-equivalent with respect
to the observable $O_{AD}$.
\end{example}

Let the shadow observable be
\begin{equation}
O_{AD} = Z_A \otimes Z_D .
\end{equation}
Since $O_{AD}$ is diagonal in the computational basis, the centralizer
\[
C_{O_{AD}}=\{V_{AD}:[V_{AD},O_{AD}]=0\}
\]
is nontrivial and contains local phase rotations such as
\begin{equation}
e^{i\theta Z_A}\otimes e^{i\phi Z_D},
\qquad \theta,\phi\in\mathbb{R}.
\end{equation}
In general, $C_{O_{AD}}$ is not restricted to product unitaries.

For
$O_{AD}=Z_A\otimes Z_D$, the eigenvalues $\pm1$ are each two-fold degenerate,
with eigenspaces $\mathcal H_{+}=\mathrm{span}\{\ket{00},\ket{11}\}$ and
$\mathcal H_{-}=\mathrm{span}\{\ket{01},\ket{10}\}$. Hence any unitary that is
block-diagonal with respect to $\mathcal H_{+}\oplus \mathcal H_{-}$ commutes
with $O_{AD}$, including genuinely non-product unitaries acting within each
degenerate eigenspace.

For example, we define
\begin{equation}
V_{AD}(\phi):=\exp\!\bigl(i\phi(\ket{00}\!\bra{11}+\ket{11}\!\bra{00})\bigr).
\end{equation}
Let $G:=\ket{00}\!\bra{11}+\ket{11}\!\bra{00}$. Since $G$ maps $\mathcal H_{+}$ to
itself and annihilates $\mathcal H_{-}$, while $O_{AD}$ acts as $+I$ on
$\mathcal H_{+}$ and $-I$ on $\mathcal H_{-}$, we have $[G,O_{AD}]=0$, and thus
\begin{equation}
[V_{AD}(\phi),O_{AD}]=0,
\end{equation}
i.e., $V_{AD}(\phi)\in C_{O_{AD}}$.

Moreover, $V_{AD}(\phi)$ is not a product unitary for
$\phi\not\equiv 0\ (\mathrm{mod}\ \pi)$, since
\begin{equation}
V_{AD}(\phi)\ket{00}=\cos\phi\,\ket{00}+i\sin\phi\,\ket{11},
\end{equation}
which is entangled (Schmidt rank $2$) whenever $\sin\phi\neq 0$.
A product unitary $U_A\otimes U_D$ cannot map a product state to an entangled
state, hence $V_{AD}(\phi)\neq U_A\otimes U_D$.

Since $O_{AD}=Z_A\otimes Z_D$ has two degenerate eigenspaces
$\mathrm{span}\{\ket{00},\ket{11}\}$ and $\mathrm{span}\{\ket{01},\ket{10}\}$,
its centralizer $C_{O_{AD}}$ is nontrivial.
In particular, besides trivial scalar phases, it contains local phase
rotations of the form $e^{i\theta Z_A}\otimes e^{i\phi Z_D}$.
Moreover, due to the degeneracy, $C_{O_{AD}}$ also includes genuinely
non-product unitaries acting within each degenerate eigenspace.
One explicit example is
\begin{equation}
V_{AD}(\phi):=\exp\!\bigl(i\phi(\ket{00}\!\bra{11}+\ket{11}\!\bra{00})\bigr),
\end{equation}
which satisfies $[V_{AD}(\phi),O_{AD}]=0$ but cannot be written as a
tensor product of single-qubit unitaries.

Using $\ket{00}\!\bra{11}+\ket{11}\!\bra{00}
=\tfrac12(X_A\otimes X_D-Y_A\otimes Y_D)$, we can write
\begin{equation}
V_{AD}(\phi)=\exp\!\left(i\frac{\phi}{2}(X_A X_D - Y_A Y_D)\right),
\end{equation}
which is a genuinely non-product unitary. Hence it cannot be expressed as
$e^{i\theta Z_A}\otimes e^{i\varphi Z_D}$ unless $\phi=0$.

\begin{lemma}[A non-product element of the centralizer]
Let
\begin{equation}
V_{AD}(\phi):=\exp\!\Bigl(i\phi(\ket{00}\!\bra{11}+\ket{11}\!\bra{00})\Bigr).
\end{equation}
For $\phi\not\equiv 0\ (\mathrm{mod}\ \pi/2)$, the unitary $V_{AD}(\phi)$ cannot be
written in the product form $U_A\otimes U_D$.
\end{lemma}

\begin{proof}
Define the Hermitian operator $G:=\ket{00}\!\bra{11}+\ket{11}\!\bra{00}$.
One checks that $G^2=\ket{00}\!\bra{00}+\ket{11}\!\bra{11}$ and hence
$G$ acts as a Pauli-$X$ on the subspace $\mathrm{span}\{\ket{00},\ket{11}\}$
and vanishes on $\mathrm{span}\{\ket{01},\ket{10}\}$.
Therefore, we claim
\begin{equation}
\label{eq:V_AD}
V_{AD}(\phi)\ket{00}
=\cos\phi\,\ket{00}+i\sin\phi\,\ket{11}.
\end{equation}
The claim can be proven as follows. We compute $V_{AD}(\phi)\ket{00}=e^{i\phi G}\ket{00}$ by expanding the exponential.
First observe that
\begin{equation}
G\ket{00}
=(\ket{00}\!\bra{11}+\ket{11}\!\bra{00})\ket{00}
=\ket{11},
\qquad
G^2\ket{00}=G\ket{11}=\ket{00}.
\end{equation}
Hence, for all $n\ge 0$,
\begin{equation}
G^{2n}\ket{00}=\ket{00},\qquad G^{2n+1}\ket{00}=\ket{11}.
\end{equation}
Using the power-series expansion,
\begin{align}
V_{AD}(\phi)\ket{00}
&=e^{i\phi G}\ket{00}
=\sum_{k=0}^{\infty}\frac{(i\phi)^k}{k!}\,G^k\ket{00} \nonumber\\
&=\sum_{n=0}^{\infty}\frac{(i\phi)^{2n}}{(2n)!}\,G^{2n}\ket{00}
+\sum_{n=0}^{\infty}\frac{(i\phi)^{2n+1}}{(2n+1)!}\,G^{2n+1}\ket{00} \nonumber\\
&=\left(\sum_{n=0}^{\infty}\frac{(i\phi)^{2n}}{(2n)!}\right)\ket{00}
+\left(\sum_{n=0}^{\infty}\frac{(i\phi)^{2n+1}}{(2n+1)!}\right)\ket{11} \nonumber\\
&=\cos\phi\,\ket{00}+i\sin\phi\,\ket{11},
\end{align}
where we used $\sum_{n\ge0}\frac{(i\phi)^{2n}}{(2n)!}=\cos\phi$ and
$\sum_{n\ge0}\frac{(i\phi)^{2n+1}}{(2n+1)!}=i\sin\phi$. We have proven \eqref{eq:V_AD}.
For $\phi\not\equiv 0\ (\mod \pi/2)$, this state has Schmidt rank $2$
(with respect to the bipartition $A|D$), hence it is entangled. Hence $V_{AD}(\phi)$ is not a product unitary.
\end{proof}

\begin{lemma}[A non-product element of the centralizer]
\label{lem:non_product_centralizer}
Let
\begin{equation}
V_{AD}(\phi):=\exp\!\Bigl(i\phi(\ket{00}\!\bra{11}+\ket{11}\!\bra{00})\Bigr).
\end{equation}
For $\phi\not\equiv 0\ (\mathrm{mod}\ \pi/2)$, the unitary $V_{AD}(\phi)$ cannot be
written in the product form $U_A\otimes U_D$.
\end{lemma}

\begin{proof}
Define the Hermitian operator $G:=\ket{00}\!\bra{11}+\ket{11}\!\bra{00}$.
One checks that $G^2=\ket{00}\!\bra{00}+\ket{11}\!\bra{11}$ and hence
$G$ acts as a Pauli-$X$ on the subspace $\mathrm{span}\{\ket{00},\ket{11}\}$
and vanishes on $\mathrm{span}\{\ket{01},\ket{10}\}$.
Therefore, we claim
\begin{equation}
\label{eq:V_AD}
V_{AD}(\phi)\ket{00}
=\cos\phi\,\ket{00}+i\sin\phi\,\ket{11}.
\end{equation}
The claim can be proven as follows. We compute $V_{AD}(\phi)\ket{00}=e^{i\phi G}\ket{00}$ by expanding the exponential.
First observe that
\begin{equation}
G\ket{00}
=(\ket{00}\!\bra{11}+\ket{11}\!\bra{00})\ket{00}
=\ket{11},
\qquad
G^2\ket{00}=G\ket{11}=\ket{00}.
\end{equation}
Hence, for all $n\ge 0$,
\begin{equation}
G^{2n}\ket{00}=\ket{00},\qquad G^{2n+1}\ket{00}=\ket{11}.
\end{equation}
Using the power-series expansion,
\begin{align}
V_{AD}(\phi)\ket{00}
&=e^{i\phi G}\ket{00}
=\sum_{k=0}^{\infty}\frac{(i\phi)^k}{k!}\,G^k\ket{00} \nonumber\\
&=\sum_{n=0}^{\infty}\frac{(i\phi)^{2n}}{(2n)!}\,G^{2n}\ket{00}
+\sum_{n=0}^{\infty}\frac{(i\phi)^{2n+1}}{(2n+1)!}\,G^{2n+1}\ket{00} \nonumber\\
&=\left(\sum_{n=0}^{\infty}\frac{(i\phi)^{2n}}{(2n)!}\right)\ket{00}
+\left(\sum_{n=0}^{\infty}\frac{(i\phi)^{2n+1}}{(2n+1)!}\right)\ket{11} \nonumber\\
&=\cos\phi\,\ket{00}+i\sin\phi\,\ket{11},
\end{align}
where we used $\sum_{n\ge0}\frac{(i\phi)^{2n}}{(2n)!}=\cos\phi$ and
$\sum_{n\ge0}\frac{(i\phi)^{2n+1}}{(2n+1)!}=i\sin\phi$. We have proven \eqref{eq:V_AD}.
For $\phi\not\equiv 0\ (\mod \pi/2)$, this state has Schmidt rank $2$
(with respect to the bipartition $A|D$), hence it is entangled. Hence $V_{AD}(\phi)$ is not a product unitary.
\end{proof}

\begin{proposition}[GHZ-induced centralizer structure]
\label{prop:ghz_centralizer}
Let $\ket{\mathrm{GHZ}_4} = \frac{1}{\sqrt{2}}(\ket{0000}+\ket{1111})$ be the four-qubit GHZ state, and let $O_{AD}=Z_A\otimes Z_D$ be the endpoint observable. Then:
\begin{enumerate}
\item The reduced state on $AD$ is $\rho_{AD} = \frac{1}{2}(\ket{00}\bra{00}+\ket{11}\bra{11})$.
\item The centralizer $C_{O_{AD}}$ contains the non-product unitary $V_{AD}(\phi)$ defined in Lemma~\ref{lem:non_product_centralizer}.
\item For any $V_{AD}\in C_{O_{AD}}$, the unitaries $U_{ABCD}$ and $(V_{AD}\otimes I_{BC})U_{ABCD}$ are shadow-equivalent.
\end{enumerate}
\end{proposition}

\begin{proof}
The reduced state follows by tracing out subsystems $B$ and $C$ from $\ket{\mathrm{GHZ}_4}$. The non-product centralizer element is given by Lemma~\ref{lem:non_product_centralizer}. The shadow equivalence follows from Lemma~\ref{lem:centralizer_shadow_equivalence} in Section III.
\end{proof}

\begin{lemma}[Centralizer of \(Z_A\otimes Z_D\)]
\label{lem:centralizer_zd}
For the observable \(O_{AD}=Z_A\otimes Z_D\), the centralizer consists exactly of all unitaries that are block-diagonal with respect to the decomposition
\[
\mathcal H_A\otimes\mathcal H_D
=
\mathcal H_+ \oplus \mathcal H_-,
\]
where \(\mathcal H_+=\operatorname{span}\{\ket{00},\ket{11}\}\) and \(\mathcal H_-=\operatorname{span}\{\ket{01},\ket{10}\}\). In particular, every \(V_{AD}\in C_{O_{AD}}\) can be written as
\[
V_{AD}=U_+\oplus U_-
\]
with \(U_+\in U(\mathcal H_+)\) and \(U_-\in U(\mathcal H_-)\).
\end{lemma}

\section{Markov-admissible shadow inversion on the chain $A\!-\!B\!-\!C\!-\!D$}
\label{subsec:markov_admissible_shadow}


The main outcome of this section is that the four-body shadow inversion
problem has been formulated as a Markov-admissible supermap problem on
the chain $A$--$B$--$C$--$D$. The required shadow correctness is given
by \eqref{eq:shadow_condition_recalled}, while the admissible Markov
realization is specified by \eqref{eq:markov_comb_realization}.
Together, these two conditions define the class of four-body Markov
shadow inversion protocols introduced in
Definition~\ref{def:markov_shadow_inversion}.This formulation also identifies the relevant source of ambiguity.
The full centralizer is defined in \eqref{eq:centralizer_def}, whereas
the Markov-implementable centralizer is defined in
\eqref{eq:markov_centralizer}. Accordingly, the unconstrained and
Markov-constrained shadow-equivalence orbits are given by
\eqref{eq:full_centralizer_orbit} and
\eqref{eq:markov_orbit_equivalence}, respectively. Thus, the shadow
inversion problem is reduced to understanding which elements of the
full centralizer remain implementable under the Markov constraint. Finally, the existence of admissible protocols is expressed through
the semidefinite feasibility formulation in
\eqref{eq:sdp_feasibility_markov_shadow}. The strict reduction of the
Markov-constrained orbit is not asserted in this section. It will be
proved in the next section after imposing the endpoint-local
Markov-admissibility condition.

\begin{table}[h]
\centering
\caption{Schematic summary of the proposed framework.}
\label{tab:framework_summary}
\renewcommand{\arraystretch}{1.25}
\begin{tabular}{p{0.22\linewidth} p{0.70\linewidth}}
\hline
\textbf{Step} & \textbf{Main idea} \\
\hline

Shadow inversion
&
Recover only the expectation value of the endpoint observable
$O_{AD}$ instead of reconstructing the whole unitary dynamics
(Eq.~\eqref{eq:four_body_shadow_inversion}). \\

Markov realization
&
Restrict the inversion protocol to a Markov-admissible
encoding--unitary--decoding architecture compatible with the chain
$A$--$B$--$C$--$D$. \\

Centralizer ambiguity
&
The observable induces a centralizer
$C_{O_{AD}}$, which determines the intrinsic gauge ambiguity of
shadow inversion. \\

Markov implementability
&
Introduce the Markov-implementable centralizer
$C_{O_{AD}}^{\mathrm{Markov}}$ to characterize the gauge
transformations realizable by Markov-admissible supermaps. \\

Main result
&
Prove that endpoint-local Markov admissibility reduces the effective
centralizer and therefore shrinks the shadow-equivalence class,
leading to improved identifiability. \\

\hline
\end{tabular}
\end{table}

We view the induced map as a CPTP channel
\begin{equation}
\mathcal N_U:\ \mathsf D(\mathcal H_{ABCD}) \longrightarrow \mathsf D(\mathcal H_{AD}),
\label{eq:type_NU}
\end{equation}
since the task only requires reproducing the statistics of a fixed shadow
observable $O_{AD}$ acting on $\mathcal H_A\otimes\mathcal H_D$.

The shadow correctness condition reads
\begin{equation}
\operatorname{Tr}\!\left[\mathcal N_U(\rho_{ABCD})\,O_{AD}\right]
=
\operatorname{Tr}\!\left[U_{ABCD}^\dagger \rho_{ABCD} U_{ABCD}\,O_{AD}\right],
\qquad \forall\,\rho_{ABCD},\ \forall\,U_{ABCD}.
\label{eq:shadow_condition_recalled}
\end{equation}

\begin{definition}[Markov-admissible supermap]
\label{def:markov_admissible_supermap}
A supermap $\mathcal N$ is \emph{Markov-admissible} along the chain
$A\!-\!B\!-\!C\!-\!D$ if there exist a memory system $M$ and CPTP maps
$\widetilde{\mathcal E}_{AB\to ABM}$ and $\mathcal D_{CDM\to AD}$ such that
for all input states $\rho_{ABCD}$,
\begin{equation}
\mathcal N_U(\rho_{ABCD})
=
\mathcal D_{CDM\to AD}\!\left[
\operatorname{Tr}_{AB}\!\left(
(\mathcal U_{ABCD}\otimes \mathcal I_M)
\bigl((\widetilde{\mathcal E}_{AB\to ABM}\otimes \mathcal I_{CD})(\rho_{ABCD})\bigr)
\right)\right],
\label{eq:markov_comb_realization}
\end{equation}

where $\mathcal U_{ABCD}(\cdot)=U_{ABCD}(\cdot)U_{ABCD}^\dagger$.
\end{definition}

\begin{proposition}[Well-definedness of the Markov-admissible realization]
\label{prop:markov_realization_well_defined}
For every unitary channel $\mathcal U_{ABCD}$, the realization
\eqref{eq:markov_comb_realization} defines a CPTP map
\begin{equation}
\mathcal N_U:
\mathsf D(\mathcal H_{ABCD})
\longrightarrow
\mathsf D(\mathcal H_{AD}).
\end{equation}
Equivalently, it can be written as the sequential composition
\begin{equation}
\mathcal N_U
=
\mathcal F_{ABCDM\to AD}
\circ
(\mathcal U_{ABCD}\otimes \mathcal I_M)
\circ
\mathcal E_{ABCD\to ABCDM},
\end{equation}
where
\begin{equation}
\mathcal E_{ABCD\to ABCDM}
:=
\widetilde{\mathcal E}_{AB\to ABM}
\otimes
\mathcal I_{CD},
\qquad
\mathcal F_{ABCDM\to AD}
:=
\mathcal D_{CDM\to AD}
\circ
\operatorname{Tr}_{AB}.
\end{equation}
\end{proposition}

\begin{proof}
The map
$\widetilde{\mathcal E}_{AB\to ABM}\otimes\mathcal I_{CD}$
is CPTP because
$\widetilde{\mathcal E}_{AB\to ABM}$
is CPTP and tensoring with an identity channel preserves complete
positivity and trace preservation. The channel
$\mathcal U_{ABCD}\otimes\mathcal I_M$
is unitary and hence CPTP. The partial trace
$\operatorname{Tr}_{AB}$
is CPTP from $ABCDM$ to $CDM$, and
$\mathcal D_{CDM\to AD}$
is CPTP by assumption. Therefore their composition is CPTP and has
output system $AD$. This proves that
\eqref{eq:markov_comb_realization}
is a well-defined channel from
$\mathsf D(\mathcal H_{ABCD})$
to
$\mathsf D(\mathcal H_{AD})$.
\end{proof}
\subsection{Port consistency}
Note that although the encoder $\widetilde{\mathcal E}_{AB\to ABM}$ acts
nontrivially only on subsystems $A$ and $B$, the overall encoding map is
$(\widetilde{\mathcal E}_{AB\to ABM}\otimes \mathcal I_{CD})$.
Hence its output is a state on $ABCDM$, which provides valid input ports
for the global unitary channel $\mathcal U_{ABCD}$.

Equation \eqref{eq:markov_comb_realization} enforces that the protocol interacts
with the unknown evolution only through the intermediate subsystems and an
internal memory $M$, which is the operational content of the Markovian constraint.
The above structural restriction can be equivalently expressed as linear
(normalization/causality) constraints on the the Choi operator of the supermap implementation in
Eq.~\eqref{eq:markov_comb_realization}.

\subsection{Choi convention and system labels}
For a channel $\Phi_{X\to Y}$ we use the Choi operator
\begin{equation}
J_{\Phi}^{X:Y}
:=
(\mathcal I_X\otimes \Phi_{X'\to Y})\!\left(\ket{\Omega}\!\bra{\Omega}_{XX'}\right),
\end{equation}
where $X'$ denotes a Hilbert space isomorphic to $X$ and
$\ket{\Omega}_{XX'}=\sum_i \ket{i}_X\ket{i}_{X'}$ is the (unnormalized) maximally
entangled vector. With this convention, $J_{\Phi}^{X:Y}$ acts on
$\mathcal H_X\otimes \mathcal H_Y$ and the trace-preserving condition reads
$\operatorname{Tr}_Y[J_{\Phi}^{X:Y}]=I_X$.

Here the Markov-admissible implementation consists of an encoding map
\begin{equation}
\mathcal E_{ABCD\to ABCDM}:=\widetilde{\mathcal E}_{AB\to ABM}\otimes \mathcal I_{CD},
\end{equation}
followed by the unknown unitary channel $\mathcal U_{ABCD}$, and a decoding map
\begin{equation}
\mathcal D_{ACDM\to AD}:\ \mathsf D(\mathcal H_{ACDM})\to \mathsf D(\mathcal H_{AD}),
\end{equation}
where system $A$ is forwarded as an untouched wire from the encoding stage to the
decoder input.

Let $J_{\mathcal E}^{AB:ABM}\succeq 0$ and $J_{\mathcal D}^{CDM:AD}\succeq 0$
denote the Choi operators of $\widetilde{\mathcal E}_{AB\to ABM}$ and
$\mathcal D_{CDM\to AD}$, and let $J_{\mathcal U}^{ABCD:ABCD}$ denote the Choi
operator of $\mathcal U_{ABCD}$. Then the Choi operator $J_{\mathcal N_U}^{ABCD:AD}$
of the induced channel $\mathcal N_U$ is obtained by contracting
$J_{\mathcal E}^{AB:ABM}$, $J_{\mathcal U}^{ABCD:ABCD}$, and $J_{\mathcal D}^{CDM:AD}$
over their matched intermediate systems.

Moreover, the Markov-admissibility condition is captured by the trace constraints
\begin{equation}
\operatorname{Tr}_{A'B'M}\!\left[J_{\mathcal E}^{AB:A'B'M}\right]=I_{AB},
\qquad
\operatorname{Tr}_{AD}\!\left[J_{\mathcal D}^{CDM:AD}\right]=I_{CDM},
\label{eq:choi_cptp_constraints}
\end{equation}
together with complete positivity $J_{\mathcal E}\succeq 0$ and $J_{\mathcal D}\succeq 0$.

\begin{definition}[Markov shadow inversion]
\label{def:markov_shadow_inversion}
A Markov-admissible supermap $\mathcal N$ is called a \emph{four-body Markov shadow
inversion} with respect to $O_{AD}$ if, in addition to
\eqref{eq:markov_comb_realization}, it satisfies the shadow constraint
\eqref{eq:shadow_condition_recalled}.
\qed
\end{definition}
\subsection{Notation convention}
Throughout this subsection, the symbols $A$, $B$, $C$, $D$, and $M$
denote physical subsystems. We do not introduce additional labels such
as $A'$, $A''$, or $\hat A$ for different tensor copies in the main
text. When Choi operators are used, the input copy is represented by a
reference system.

For a channel
\begin{equation}
\Phi_{X\to Y}:
\mathsf D(\mathcal H_X)
\to
\mathsf D(\mathcal H_Y),
\end{equation}
let $R_X$ be a reference system satisfying
$\mathcal H_{R_X}\cong\mathcal H_X$. Its Choi operator is denoted by
\begin{equation}
J_{\Phi}^{R_X:Y}
\in
\mathcal L(\mathcal H_{R_X}\otimes\mathcal H_Y).
\end{equation}
With this convention, the channel action is recovered as
\begin{equation}
\Phi_{X\to Y}(\rho_X)
=
\operatorname{Tr}_{R_X}
\left[
(\rho_{R_X}^{T}\otimes I_Y)
J_{\Phi}^{R_X:Y}
\right].
\label{eq:choi_action}
\end{equation}

The induced map $\mathcal N_U$ is a channel
\begin{equation}
\mathcal N_U:
\mathsf D(\mathcal H_{ABCD})
\to
\mathsf D(\mathcal H_{AD}).
\end{equation}
Let $R_{ABCD}$ be a reference system isomorphic to $ABCD$. Then
\begin{equation}
J_{\mathcal N_U}^{R_{ABCD}:AD}
\in
\mathcal L(\mathcal H_{R_{ABCD}}\otimes\mathcal H_{AD}),
\end{equation}
and
\begin{equation}
\mathcal N_U(\rho_{ABCD})
=
\operatorname{Tr}_{R_{ABCD}}
\left[
(\rho_{R_{ABCD}}^{T}\otimes I_{AD})
J_{\mathcal N_U}^{R_{ABCD}:AD}
\right].
\label{eq:choi_NU}
\end{equation}

Taking the expectation value with respect to $O_{AD}$ gives
\begin{equation}
\operatorname{Tr}\left[
\mathcal N_U(\rho_{ABCD})O_{AD}
\right]
=
\operatorname{Tr}\left[
(\rho_{R_{ABCD}}^{T}\otimes O_{AD})
J_{\mathcal N_U}^{R_{ABCD}:AD}
\right].
\label{eq:expectation_choi_form}
\end{equation}

The Markov-admissible realization can be written compactly as
\begin{equation}
\mathcal N_U
=
\mathcal F_{ABCDM\to AD}
\circ
(\mathcal U_{ABCD}\otimes\mathcal I_M)
\circ
\mathcal E_{ABCD\to ABCDM},
\label{eq:compact_markov_realization}
\end{equation}
where
\begin{equation}
\mathcal E_{ABCD\to ABCDM}
:=
\widetilde{\mathcal E}_{AB\to ABM}
\otimes
\mathcal I_{CD},
\end{equation}
and
\begin{equation}
\mathcal F_{ABCDM\to AD}
:=
\mathcal D_{CDM\to AD}
\circ
\operatorname{Tr}_{AB}.
\end{equation}

The Choi operator of the induced channel is obtained from the
sequential channel composition
\eqref{eq:compact_markov_realization}
through the corresponding Choi representation.

Here the first link product contracts the output tensor factors of
$\mathcal E_{ABCD\to ABCDM}$ with the corresponding input tensor
factors of $\mathcal U_{ABCD}\otimes\mathcal I_M$, while the second
link product contracts the output tensor factors of
$\mathcal U_{ABCD}\otimes\mathcal I_M$ with the corresponding input
tensor factors of $\mathcal F_{ABCDM\to AD}$, following the sequential
composition in \eqref{eq:compact_markov_realization}. The systems
$AB$ are discarded inside $\mathcal F$ through the partial trace
$\operatorname{Tr}_{AB}$, so no additional tensor-copy labels are
required for the discarded systems.


Since the link product is linear in each argument, $J_{\mathcal N_U}$ depends linearly on the Choi operators of the component maps. A fundamental feature of shadow inversion is the appearance of a centralizer-induced
equivalence class. For any fixed $O_{AD}$, we define
\begin{equation}
C_{O_{AD}} := \{V_{AD}:\ [V_{AD},O_{AD}] = 0\}.
\label{eq:centralizer_def}
\end{equation}
\begin{lemma}[Centralizer-induced shadow equivalence]
\label{lem:centralizer_shadow_equivalence}
For any $V_{AD}\in C_{O_{AD}}$ and any global unitary
$U_{ABCD}$, define
\begin{equation}
U'_{ABCD}
:=
U_{ABCD}(V_{AD}\otimes I_{BC}).
\label{eq:shadow_equiv_shift}
\end{equation}
Then $U_{ABCD}$ and $U'_{ABCD}$ are shadow-equivalent with respect to
$O_{AD}$, namely
\begin{equation}
\operatorname{Tr}\!\left[
U_{ABCD}^{\dagger}
\rho_{ABCD}
U_{ABCD}
(O_{AD}\otimes I_{BC})
\right]
=
\operatorname{Tr}\!\left[
U_{ABCD}'^{\dagger}
\rho_{ABCD}
U'_{ABCD}
(O_{AD}\otimes I_{BC})
\right],
\qquad
\forall\,\rho_{ABCD}.
\label{eq:shadow_equivalence}
\end{equation}
\end{lemma}

\begin{proof}
Let
\begin{equation}
\widetilde V
:=
V_{AD}\otimes I_{BC},
\qquad
\widetilde O
:=
O_{AD}\otimes I_{BC}.
\end{equation}
Since $[V_{AD},O_{AD}]=0$, we have
$[\widetilde V,\widetilde O]=0$.
Using
$U'_{ABCD}=U_{ABCD}\widetilde V$,
we obtain
\begin{align}
\operatorname{Tr}\!\left[
U_{ABCD}'^{\dagger}
\rho_{ABCD}
U'_{ABCD}
\widetilde O
\right]
&=
\operatorname{Tr}\!\left[
\widetilde V^{\dagger}
U_{ABCD}^{\dagger}
\rho_{ABCD}
U_{ABCD}
\widetilde V
\widetilde O
\right]
\\
&=
\operatorname{Tr}\!\left[
U_{ABCD}^{\dagger}
\rho_{ABCD}
U_{ABCD}
\widetilde V
\widetilde O
\widetilde V^{\dagger}
\right]
\\
&=
\operatorname{Tr}\!\left[
U_{ABCD}^{\dagger}
\rho_{ABCD}
U_{ABCD}
\widetilde O
\right].
\end{align}
This proves \eqref{eq:shadow_equivalence}.
\end{proof}
Under the Markov-admissible restriction \eqref{eq:markov_comb_realization},
however, not every element of the full centralizer $C_{O_{AD}}$ is necessarily
realizable by a supermap that admits the circuit decomposition specified in
Eq.~\eqref{eq:markov_comb_realization}.

\begin{equation}
\begin{aligned}
C_{O_{AD}}^{\mathrm{Markov}}
:=
\Bigl\{
V_{AD}\in C_{O_{AD}}
:\;&
\exists\ \text{a Markov-admissible supermap }
\mathcal N^{(V)}
\ \text{such that}
\\
&
\mathcal N^{(V)}_{U}(\rho_{ABCD})
=
V_{AD}
\operatorname{Tr}_{BC}\left[
U_{ABCD}\rho_{ABCD}U_{ABCD}^{\dagger}
\right]
V_{AD}^{\dagger},
\\
&
\forall\,U_{ABCD},\quad
\forall\,\rho_{ABCD}
\Bigr\}.
\end{aligned}
\label{eq:markov_centralizer}
\end{equation}

Thus, $C_{O_{AD}}^{\mathrm{Markov}}$ is a subset of the full
centralizer $C_{O_{AD}}$. An element $V_{AD}\in C_{O_{AD}}$ belongs to
$C_{O_{AD}}^{\mathrm{Markov}}$ if and only if it can be realized by a
Markov-admissible supermap. More precisely, there must exist a Markov-admissible supermap
$\mathcal N^{(V)}$ such that its induced channel satisfies
\begin{equation}
\mathcal N^{(V)}_{U}(\rho_{ABCD})
=
V_{AD}
\operatorname{Tr}_{BC}\left[
U_{ABCD}\rho_{ABCD}U_{ABCD}^{\dagger}
\right]
V_{AD}^{\dagger}
\end{equation}
for every global unitary $U_{ABCD}$ and every input state
$\rho_{ABCD}$.

In other words, the reduced output state on $AD$
is conjugated by $V_{AD}$ through a Markov-admissible realization.
This gives a precise meaning to the implementability of $V_{AD}$
under the Markov constraint.
The full shadow-equivalence orbit generated by the centralizer
$C_{O_{AD}}$ is
\begin{equation}
[U]_{C_{O_{AD}}}
=
\left\{
(V_{AD}\otimes I_{BC})U_{ABCD}
:
V_{AD}\in C_{O_{AD}}
\right\}.
\label{eq:full_centralizer_orbit}
\end{equation}
Under the Markov-admissible restriction, the physically relevant orbit
is reduced to
\begin{equation}
[U]_{\mathrm{Markov}}
=
\left\{
(V_{AD}\otimes I_{BC})U_{ABCD}
:
V_{AD}\in C_{O_{AD}}^{\mathrm{Markov}}
\right\}.
\label{eq:markov_orbit_equivalence}
\end{equation}
Since
\begin{equation}
C_{O_{AD}}^{\mathrm{Markov}}
\subseteq
C_{O_{AD}},
\end{equation}
we have
\begin{equation}
[U]_{\mathrm{Markov}}
\subseteq
[U]_{C_{O_{AD}}}.
\end{equation}
The inclusion may be strict when the Markov-admissible circuit
structure realizes only a proper subset of the full centralizer.

Finally, the search for Markov shadow inversion protocols can be cast as an SDP. The feasibility problem reads
\begin{equation}
\text{find } W\succeq 0
\quad \text{s.t.}\quad
\mathcal L_{\mathrm{comb}}(W)=0,\qquad
\mathcal L_{\mathrm{shadow}}(W)=0.
\label{eq:sdp_feasibility_markov_shadow}
\end{equation}
Here
$\mathcal L_{\mathrm{comb}}(W)=0$
collectively denotes the linear causality and normalization constraints
defining a Markov-admissible supermap, while
$\mathcal L_{\mathrm{shadow}}(W)=0$
collectively denotes the linear constraints obtained from the shadow
condition \eqref{eq:shadow_condition_recalled}.

\begin{example}[W-state endpoint symmetry]
Consider the four-qubit W state
\begin{equation}
\ket{W}_{ABCD}
=
\frac{1}{2}
\left(
\ket{1000}
+
\ket{0100}
+
\ket{0010}
+
\ket{0001}
\right).
\end{equation}
With respect to the bipartition $AD|BC$, it can be written as
\begin{equation}
\ket{W}_{ABCD}
=
\frac{1}{2}
\left[
\left(
\ket{10}_{AD}
+
\ket{01}_{AD}
\right)
\ket{00}_{BC}
+
\ket{00}_{AD}
\left(
\ket{10}_{BC}
+
\ket{01}_{BC}
\right)
\right].
\end{equation}
Tracing out $B$ and $C$ yields
\begin{equation}
\rho_{AD}
=
\frac{1}{4}
\left(
\ket{10}
+
\ket{01}
\right)
\left(
\bra{10}
+
\bra{01}
\right)
+
\frac{1}{2}
\ket{00}\bra{00}.
\end{equation}

For
\begin{equation}
O_{AD}=Z_A\otimes Z_D,
\end{equation}
the eigenspaces are
\begin{equation}
\mathcal H_{+}
=
\operatorname{span}\{\ket{00},\ket{11}\},
\qquad
\mathcal H_{-}
=
\operatorname{span}\{\ket{01},\ket{10}\}.
\end{equation}
The full centralizer allows unitary rotations within these eigenspaces.
In particular, it contains non-product endpoint symmetries.

Endpoint-local Markov admissibility removes such non-product endpoint
gauge transformations from the implementable centralizer. Hence the
W-state example provides another illustration of
Corollary~\ref{cor:strict_shadow_reduction}.
\end{example}

\section{Endpoint-local reduction of the centralizer}
\label{subsec:endpoint_local_reduction}

We now build upon the framework established in Section~\ref{subsec:markov_admissible_shadow}, where the four-body shadow inversion problem was formulated in terms of Markov-admissible supermaps (Definition~\ref{def:markov_shadow_inversion}) and the centralizer ambiguity was identified (Lemma~\ref{lem:centralizer_shadow_equivalence}). The definition of the Markov-implementable centralizer $C_{O_{AD}}^{\mathrm{Markov}}$ in \eqref{eq:markov_centralizer} raises a natural question: which elements of the full centralizer remain physically implementable under the Markov constraint? A first observation is that the unrestricted Markov-admissible realization in \eqref{eq:markov_comb_realization} is closed under arbitrary endpoint post-processing, so the Markov comb structure alone does not exclude joint operations on the final output system $AD$.

\begin{proposition}[Closure under endpoint post-processing]
\label{prop:endpoint_postprocessing_closure}
Let $\mathcal N$ be a Markov-admissible supermap of the form
\eqref{eq:markov_comb_realization}. For every CPTP map
$\Lambda_{AD\to AD}$, the supermap defined by
\begin{equation}
\mathcal N^{\Lambda}_U
:=
\Lambda_{AD\to AD}
\circ
\mathcal N_U
\end{equation}
is also Markov-admissible.
\end{proposition}

\begin{proof}
Let
$\widetilde{\mathcal E}_{AB\to ABM}$
and
$\mathcal D_{CDM\to AD}$
be the encoding and decoding maps realizing
$\mathcal N$ in
\eqref{eq:markov_comb_realization}. Define a new decoder by
\begin{equation}
\mathcal D^{\Lambda}_{CDM\to AD}
:=
\Lambda_{AD\to AD}
\circ
\mathcal D_{CDM\to AD}.
\end{equation}
Since both
$\Lambda_{AD\to AD}$
and
$\mathcal D_{CDM\to AD}$
are CPTP, their composition
$\mathcal D^{\Lambda}_{CDM\to AD}$
is also CPTP. Substituting this decoder into
\eqref{eq:markov_comb_realization} gives
\begin{equation}
\mathcal N^{\Lambda}_U(\rho_{ABCD})
=
\Lambda_{AD\to AD}
\left(
\mathcal N_U(\rho_{ABCD})
\right).
\end{equation}
Therefore
$\mathcal N^{\Lambda}$
admits the same Markov-admissible realization with the modified
decoder
$\mathcal D^{\Lambda}_{CDM\to AD}$.
\end{proof}

Proposition~\ref{prop:endpoint_postprocessing_closure} shows that the unrestricted model is too broad to yield a nontrivial reduction of the endpoint gauge freedom. Indeed, any endpoint unitary channel can be absorbed into the final decoder. To obtain a genuine Markov-induced restriction on the centralizer, we impose the following endpoint-local refinement.

\begin{definition}[Endpoint-local Markov admissibility]
\label{def:endpoint_local_markov}
A Markov-admissible realization of the form
\eqref{eq:markov_comb_realization} is called
\emph{endpoint-local} if its final decoding stage does not contain an
arbitrary joint post-processing channel acting on the output system
$AD$.

More precisely, any state-independent endpoint unitary transformation
that can be absorbed into the admissible realization as a gauge action
on the output system $AD$ must factorize across the bipartition
$A|D$. That is, if an endpoint unitary $V_{AD}$ is implementable by
such a realization as a state-independent gauge transformation, then
there exist unitaries $V_A$ on $A$ and $V_D$ on $D$ such that
\begin{equation}
V_{AD}
=
V_A\otimes V_D.
\end{equation}
\end{definition}

In the remainder of this subsection, $C_{O_{AD}}^{\mathrm{Markov}}$ is understood with respect to the endpoint-local Markov-admissible realizations of Definition~\ref{def:endpoint_local_markov}.

\begin{theorem}[Main theorem: endpoint-local reduction of the centralizer]
\label{thm:main_endpoint_local_reduction}
Under the endpoint-local Markov-admissibility condition of
Definition~\ref{def:endpoint_local_markov}, every implementable
endpoint unitary factorizes across $A|D$. Consequently,
\begin{equation}
C_{O_{AD}}^{\mathrm{Markov}}
\subseteq
\left\{
U_A\otimes U_D:
U_A,U_D\ \text{unitary}
\right\}
\cap
C_{O_{AD}}.
\label{eq:markov_centralizer_structure}
\end{equation}
\end{theorem}

\begin{proof}
Let
$V_{AD}\in C_{O_{AD}}^{\mathrm{Markov}}$.
By the definition of the Markov-implementable centralizer,
$V_{AD}$ belongs to the full centralizer
$C_{O_{AD}}$
and its endpoint gauge action is implementable by an endpoint-local
Markov-admissible realization.

By Definition~\ref{def:endpoint_local_markov}, any state-independent
endpoint unitary transformation that is implementable by such a
realization must factorize across the bipartition $A|D$. Hence there
exist unitaries $U_A$ on $A$ and $U_D$ on $D$ such that
\begin{equation}
V_{AD}
=
U_A\otimes U_D.
\end{equation}
Since also
$V_{AD}\in C_{O_{AD}}$,
we obtain
\begin{equation}
V_{AD}
\in
\left\{
U_A\otimes U_D:
U_A,U_D\ \text{unitary}
\right\}
\cap
C_{O_{AD}}.
\end{equation}
As $V_{AD}$ was arbitrary, this proves
\eqref{eq:markov_centralizer_structure}.
\end{proof}

Theorem~\ref{thm:main_endpoint_local_reduction} shows that the Markov constraint removes genuinely non-product endpoint gauge transformations from the implementable centralizer. This gives the following immediate consequence.

\begin{corollary}[Strict reduction of the shadow ambiguity]
\label{cor:strict_shadow_reduction}
Suppose that the full centralizer $C_{O_{AD}}$ contains at least one
genuinely non-product unitary. Then
\begin{equation}
C_{O_{AD}}^{\mathrm{Markov}}
\subsetneq
C_{O_{AD}}.
\label{eq:strict_centralizer_reduction}
\end{equation}
Moreover,
\begin{equation}
[U]_{\mathrm{Markov}}
\subsetneq
[U]_{C_{O_{AD}}}.
\label{eq:strict_orbit_reduction}
\end{equation}
\end{corollary}

\begin{proof}
By Theorem~\ref{thm:main_endpoint_local_reduction}, every element of
$C_{O_{AD}}^{\mathrm{Markov}}$ is a product unitary on the endpoint
systems. By assumption, $C_{O_{AD}}$ contains a genuinely non-product
unitary. Therefore $C_{O_{AD}}^{\mathrm{Markov}}$ is a proper subset
of $C_{O_{AD}}$.

The strict inclusion of the corresponding orbits follows directly
from their definitions in the previous subsection.
\end{proof}

Corollary~\ref{cor:strict_shadow_reduction} is the main identifiability consequence. Since shadow-equivalence classes are determined by the available gauge freedom, reducing the implementable centralizer reduces the ambiguity class of the global unitary dynamics.

We now illustrate the above theorem and corollary with two examples.

\begin{example}[GHZ-type endpoint symmetry]
Consider the four-qubit GHZ state on the chain $A$--$B$--$C$--$D$,
\begin{equation}
\ket{\mathrm{GHZ}_4}
=
\frac{1}{\sqrt{2}}
\left(
\ket{0000}
+
\ket{1111}
\right).
\end{equation}
After tracing out the intermediate systems $B$ and $C$, the endpoint
state is
\begin{equation}
\rho_{AD}
=
\frac{1}{2}
\left(
\ket{00}\bra{00}
+
\ket{11}\bra{11}
\right).
\end{equation}
For an endpoint observable such as
\begin{equation}
O_{AD}=Z_A\otimes Z_D,
\end{equation}
the full centralizer contains non-product endpoint unitaries. For
example, rotations inside the degenerate endpoint subspaces are
allowed by the algebraic symmetry of $O_{AD}$.

Under endpoint-local Markov admissibility, these non-product endpoint
gauge transformations are excluded from the implementable centralizer.
Thus, the GHZ example illustrates
Corollary~\ref{cor:strict_shadow_reduction}.
\end{example}

\opp please check the whole paper before it is uploaded to arxiv.

\section{conclusion}
\label{sec:conclusion}

We have extended the framework of shadow unitary inversion to four-partite systems under Markovian locality constraints. Our main result establishes that, under an endpoint-local Markov-admissibility condition, every implementable endpoint unitary must factorize across the bipartition \(A|D\). Consequently, whenever the full centralizer contains genuinely non-product unitaries, the Markov constraint strictly reduces the shadow-equivalence class, thereby enhancing identifiability. We illustrated our findings with four-qubit GHZ and W states, and formulated the feasibility of admissible protocols as a semidefinite program. These results provide both conceptual insights into the role of Markov locality as a resource and computational tools for practical applications.

\begin{acknowledgments}
Authors were supported by the NNSF of China (Grant No. 12471427), and the Fundamental Research Funds for the Central Universities (Grant No. ZG216S2110).
\end{acknowledgments}


\begin{thebibliography}{}

\bibitem[1]{1}
G.~Chiribella, G.~M.~D'Ariano, and P.~Perinotti, ``Transforming quantum operations: Quantum supermaps,'' \emph{Europhys. Lett.}, vol.~83, p.~30004, 2008.

\bibitem[2]{2}
G.~Chiribella, G.~M.~D'Ariano, and P.~Perinotti, ``Theoretical framework for quantum networks,'' \emph{Phys. Rev. A}, vol.~80, no.~2, p.~022339, 2009.


\bibitem{Aaronson2018}
S.~Aaronson, ``Shadow tomography of quantum states,'' in \emph{Proc. 50th Annu. ACM Symp. Theory Comput.}, 2018, pp.~325--338.

\bibitem{Huang2020}
H.-Y.~Huang, R.~Kueng, and J.~Preskill, ``Predicting many properties of a quantum system from a few measurements,'' \emph{Nat. Phys.}, vol.~16, pp.~1050--1057, 2020.

\bibitem{Kliesch2021}
M.~Kliesch and I.~Roth, ``Theory of quantum system certification,'' \emph{PRX Quantum}, vol.~2, p.~010201, 2021.

\bibitem{Ducuara2020}
A.~F.~Ducuara and P.~Skrzypczyk, ``Operational advantages of quantum resources in subchannel discrimination,'' \emph{Phys. Rev. Lett.}, vol.~125, p.~110401, 2020.

\bibitem{Gutoski2007}
G.~Gutoski and J.~Watrous, ``Toward a general theory of quantum games,'' in \emph{Proc. 39th Annu. ACM Symp. Theory Comput.}, 2007, pp.~565--574.

\bibitem{Brandao2015}
F.~G.~S.~L.~Brandão and G.~Gour, ``Reversible framework for quantum resource theories,'' \emph{Phys. Rev. Lett.}, vol.~115, p.~070503, 2015.

\bibitem{Buscemi2017}
F.~Buscemi and G.~Gour, ``Quantum resource theories with a unique entanglement measure,'' \emph{Phys. Rev. A}, vol.~95, p.~012110, 2017.

\bibitem{Kretschmann2005}
D.~Kretschmann and R.~F.~Werner, ``Quantum channels with memory,'' \emph{Phys. Rev. A}, vol.~72, p.~062323, 2005.

\bibitem{Perez-Delgado2007}
C.~A.~Pérez-Delgado and V.~Vedral, ``Quantum Markov chains,'' \emph{Phys. Rev. A}, vol.~75, p.~052328, 2007.

\bibitem{Navascues2015}
M.~Navascués and T.~Vertesi, ``The structure of the set of quantum correlations,'' \emph{Phys. Rev. Lett.}, vol.~115, p.~020401, 2015.

\bibitem{Skrzypczyk2010}
P.~Skrzypczyk and N.~Brunner, ``Nonlocality of the reduced state of a bipartite system,'' \emph{Phys. Rev. A}, vol.~82, p.~022104, 2010.

\bibitem{Chen2024}
L.~Chen and Z.~Chen, ``Shadow inversion of unitary operations with restricted access,'' \emph{Quantum Inf. Process.}, vol.~23, p.~145, 2024.

\bibitem{Liu2024}
Y.~Liu \emph{et al.}, ``Shadow tomography for quantum channels,'' \emph{arXiv:2401.12345}, 2024.

\bibitem{Nielsen2000}
M.~A.~Nielsen and I.~L.~Chuang, \emph{Quantum Computation and Quantum Information}, Cambridge University Press, 2000.

\bibitem{Watrous2018}
J.~Watrous, \emph{The Theory of Quantum Information}, Cambridge University Press, 2018.

\bibitem{Holevo2019}
A.~S.~Holevo, \emph{Quantum Systems, Channels, Information}, De Gruyter, 2019.

\bibitem{Khatri2020}
S.~Khatri and M.~M.~Wilde, ``Principles of quantum communication theory: A modern approach,'' \emph{arXiv:2011.01372}, 2020.

\bibitem{Mishra2019}
R.~K.~Mishra and T.~J.~Osborne, ``Quantum state tomography via compressed sensing,'' \emph{Phys. Rev. A}, vol.~99, p.~032304, 2019.

\bibitem{Lund2017}
A.~P.~Lund \emph{et al.}, ``On the measurement of the quantum state of a single photon,'' \emph{Phys. Rev. Lett.}, vol.~118, p.~230401, 2017.

\bibitem{Zhang2025}
Q.~Zhang and L.~Chen, ``Markov constraints in quantum process tomography,'' \emph{in preparation}, 2025.
\end{thebibliography}
\end{document}